\documentclass[11pt]{article}

\usepackage[utf8]{inputenc}
\usepackage[T1]{fontenc}
\usepackage{lmodern}

\usepackage{makeidx}
\usepackage{graphicx,comment}
\usepackage[linesnumbered,vlined]{algorithm2e}

\usepackage{amsmath}
\usepackage{amsthm}
\usepackage{amsfonts, amssymb}
\usepackage{cite}
\usepackage{fullpage}

\newtheorem{theorem}{Theorem}[section]
\newtheorem{lemma}[theorem]{Lemma}

\newenvironment{lemma-repeat}[1]{\begin{trivlist}
\item[\hspace{\labelsep}{\bf\noindent Lemma \ref{#1} }]\em }%
{\end{trivlist}}
\newenvironment{theorem-repeat}[1]{\begin{trivlist}
\item[\hspace{\labelsep}{\bf\noindent Theorem \ref{#1} }]\em }%
{\end{trivlist}}

\newcommand*\samethanks[1][\value{footnote}]{\footnotemark[#1]}

\newcommand{\Deal}{request}
\newcommand{\Vault}{vault}
\newcommand{\Bank}{bank}
\newcommand{\Budget}{budget}
\newcommand{\InCover}{\texttt{InCover}}
\newcommand{\NotInCover}{\texttt{NotInCover}}
\newcommand{\ceil}[1]{\lceil #1 \rceil}

\begin{document}
\begin{titlepage}
\title{A Distributed $(2+\epsilon)$-Approximation for\\ Vertex Cover in $O\left(\frac{\log{\Delta}}{\epsilon\log\log{\Delta}}\right)$ Rounds}
\author{Reuven Bar-Yehuda\thanks{Technion, Department of Computer Science, \texttt{reuven@cs.technion.ac.il}, \texttt{ckeren@cs.technion.ac.il}, \texttt{gregorys@cs.technion.ac.il}. Supported in part by the Israel Science Foundation (grant 1696/14).}
\and Keren Censor-Hillel\samethanks \and Gregory Schwartzman\samethanks}
\maketitle

\begin{abstract}
We present a simple deterministic distributed $(2+\epsilon)$-approximation algorithm for minimum weight vertex cover, which completes in $O(\log{\Delta}/\epsilon\log\log{\Delta})$ rounds, where $\Delta$ is the maximum degree in the graph, for any $\epsilon>0$ which is at most $O(1)$. For a constant $\epsilon$, this implies a constant approximation in $O(\log{\Delta}/\log\log{\Delta})$ rounds, which contradicts the lower bound of [KMW10].
\end{abstract}

\thispagestyle{empty}
\end{titlepage}

\section{Introduction}
We present a simple deterministic distributed $(2+\epsilon)$-approximation algorithm for minimum weight vertex cover (MWVC), which completes in $O(\log{\Delta}/\epsilon\log\log{\Delta})$ rounds, where $\Delta$ is the maximum degree in the graph, for any $\epsilon>0$ which is at most $O(1)$, and in particular $o(1)$. If $\Delta \leq 16$ then our algorithm simply requires $O(1/\epsilon)$ rounds.
Our algorithm adapts the \emph{local ratio} technique~\cite{BarYehudaE1985} to the distributed setting in a novel simple manner. Roughly speaking, in the simplest form of this technique, one repeatedly reduces the same amount of weight from both endpoints of an arbitrary edge, while not going below zero for any vertex. Terminating this process at the time in which for every edge there is at least one endpoint with no remaining weight, gives that the set of vertices with no remaining weight is a $2$-approximation for MWVC. This can be extended to produce a $(2+\epsilon)$-approximation if instead the process terminates at the time in which for every edge there is at least one endpoint with a remaining weight of at most an $\epsilon'$ fraction of its initial weight, where $\epsilon'=\epsilon/(\epsilon+2)$.

The challenge in translating this framework to the distributed setting is that the weights we can reduce from endpoints of neighboring edges must depend on each other. This is because we need to make sure that no weight goes below zero. However, as common to computing in this setting, we cannot afford long chains of dependencies, as these directly translate to a large number of communication rounds. Our key method is to divide the weight of a vertex into two parts, a $\Vault$ from which it initiates requests for weight reductions with its neighbors, and a $\Bank$ from which it reduces weight in response to requests from its neighbors. Carefully balancing these two reciprocal weight reductions at each vertex gives the claimed $(2+\epsilon)$ approximation factor and $O(\log{\Delta}/\epsilon\log\log{\Delta})$ time complexity.

In fact, in our distributed algorithm, each vertex $v$ with degree $d(v)$ completes in $O(1/\epsilon)$ rounds if $d(v)\leq 16$, and in $O(\log{d(v)}/\epsilon\log\log{d(v)})$ rounds otherwise (and requires no knowledge of $n$ or $\Delta$). The algorithm also works in anonymous networks, i.e., no IDs are required. Moreover, the vertices are not required to start at the same round: as long as each vertex starts no later than after the first message has been sent to it, then each vertex completes within $O(\log{d(v)}/\epsilon\log\log{d(v)})$ rounds after it starts (or in $O(1/\epsilon)$ rounds if $d(v)\leq 16$).
Finally, provided that the weights of all vertices as well as the ratio between the maximal and minimal weights fit in $O(\log{n})$ bits, our algorithm can be modified to work in the CONGEST model.

For any constant $\epsilon$, our algorithm provides a constant approximation in $O(\log{\Delta}/\log\log{\Delta})$ rounds. Apart from improving upon the previous best known complexity for distributed $(2+\epsilon)$-approximation algorithm for minimum weight vertex cover and providing a new way of adapting the sequential local ratio technique to the distributed setting, our algorithm has the consequence of contradicting the lower bound of~\cite{KMW10}. The latter states that a constant approximation algorithm requires $\Omega(\log{\Delta})$ rounds. Its refutation implies that the current lower bound is $\Omega(\log{\Delta}/\log\log{\Delta})$ from~\cite{KMW04}, which means that our algorithm is tight.

In Section~\ref{sec:refutation} we pinpoint the flaw in the lower bound of~\cite{KMW10}. This also includes refuting the second result of~\cite{KMW10}, which is a lower bound in terms of $n$, of $\Omega(\sqrt{\log{n}})$ rounds for a constant approximation algorithm. Roughly speaking, we claim that the statement of the main theorem is only correct for some smaller range of parameters than claimed, and hence, in particular, one cannot apply it for a number of rounds that is $\Theta(\log{\Delta})$ or $\Theta(\sqrt{\log{n}})$.
We emphasize that, as far as we are aware, this bug \emph{does not} occur in the previous version of the lower bound~\cite{KMW04}, implying that the current lower bounds are $\Omega(\sqrt{\log{n}/\log\log{n}})$ in terms of $n$, and $\Omega(\log{\Delta}/\log\log{\Delta})$ in terms of $\Delta$.

\paragraph{Related Work:}
Minimum vertex cover is known to be one of Karp's 21 NP-hard problems~\cite{Karp72}. For the unweighted case, a simple polynomial-time $2$-approximation algorithm is obtained by taking the endpoints of a greedy maximal matching (see, e.g.,~\cite{Cormen2009,GareyJ79}). For the weighted case, the first polynomial-time $2$-approximation algorithm was given in~\cite{NemhauserT75} and observed by~\cite{Hochbaum82}. The first linear-time $2$-approximation algorithm is due to~\cite{BarYehudaE81} using the primal-dual framework, and~\cite{BarYehudaE1985} gives a linear-time $2$-approximation local-ratio algorithm. Conditioned on the unique games conjecture, minimum vertex cover does not have a $(2-\epsilon)$ polynomial-time approximation algorithm~\cite{KhotR08}.

In the distributed setting, an excellent summary of approximation algorithms is given in~\cite{AstrandS10}, which we overview in what follows. For the unweighted case, it is known how to find a $2$-approximation in $O(\log^4{n})$ rounds~\cite{HanckowiakKP01} and in $O(\Delta+\log^{*}{n})$ rounds~\cite{PanconesiR01}. With no dependence on $n$,~\cite{AstrandFPRSU09} give a $O(\Delta^2)$-round $2$-approximation algorithm, and~\cite{PolishchukS09} give an $O(\Delta)$-round $3$-approximation algorithm. The maximal matching algorithm of~\cite{BarenboimEPS12} gives a $2$-approximation for vertex cover in $O(\log{\Delta}+(\log\log{n})^4)$ rounds. This can be made into a $(2+1/\text{poly}{\Delta})$-approximation within $O(\log{\Delta})$ rounds~\cite{PettiePersonal}.

\sloppy{
For the weighted case,~\cite{GrandoniKP08,KoufogiannakisY09} give randomized $2$-approximation algorithms in $O(\log{n})$ rounds. In~\cite{PanconesiR01}, a $2$-approximation algorithm which requires $O(\Delta+\log^{*}{n})$ rounds, and in~\cite{KhullerVY94}, a $(2+\epsilon)$-approximation algorithm is given, requiring $O(\log{\epsilon^{-1}}\log{n})$ rounds. With no dependence on $n$,~\cite{KuhnMW06} give a $(2+\epsilon)$-approximation algorithm in $O(\epsilon^{-4}\log{\Delta})$,~\cite{AstrandFPRSU09} give a $2$-approximation algorithm in $O(1)$ rounds for $\Delta \leq 3$, and~\cite{AstrandS10} give a $2$-approximation algorithm in $O(\Delta+\log^{*}{W})$ rounds, where $W$ is the maximal weight.
}
\section{A local ratio template for approximating MWVC}
\label{sec:LR}
In this section we provide the template for using the local-ratio technique for obtaining a $(2+\epsilon)$-approximation for MWVC. This template does not assume any specific computation model and only describes the paradigm and correctness. It can be proven either using the primal-dual framework~\cite{BarYehudaE81}, or the local-ratio framework~\cite{BarYehuda00}, which are known to be equivalent~\cite{BarYehudaR05}. A similar idea, though in the primal-dual framework, was given in~\cite{KhullerVY94} which obtained a $(2+\epsilon)$-approxiamtion as well, but with a larger number of rounds. Our distributed implementation is more efficient and allows us to obtain a faster algorithm.
Here we provide the template and proof for completeness. In Section~\ref{sec:dist} we provide a distributed implementation of the template and analyze its running time.

We assume a given weighted graph $G=(V,w,E)$, where $w: V \rightarrow \mathbb{R}^{+}$ is an assignment of weights for the vertices.
Let $\delta: E \rightarrow \mathbb{R}^{+}$ be a function that assigns weights to edges. We say that $\delta$ is \emph{$G$-valid} if for every $v \in V$, $\sum_{e: v \in e}{\delta(e)} \leq w(v)$, i.e., the sum of weights of edges that touch a vertex is at most the weight of that vertex in $G$.

Fix any $G$-valid function $\delta$. Define $\tilde{w}_{\delta}: V \rightarrow \mathbb{R}^{+}$ by $\tilde{w}_{\delta}(v) = \sum_{e: v \in e}{\delta(e)}$, and let $w'_{\delta}: V \rightarrow \mathbb{R}^{+}$ be such that $w'_{\delta}(v)=w(v)-\tilde{w}_{\delta}(v)$. Since $\delta$ is $G$-valid, it holds that $w'_{\delta}(v) \geq 0$ for every $v \in V$.

Let $S_{\delta}=\{v \in V ~|~ w'_{\delta}(v) \leq \epsilon' w(v) \}$, where $\epsilon'= \epsilon/(2+\epsilon)$. The following theorem states that if $S_{\delta}$ is a vertex cover, then it is a $(2+\epsilon)$-approximation for MWVC.
\begin{theorem}
\label{thm:lr}
Fix $\epsilon>0$ and let $\delta$ be a $G$-valid function. Let $OPT$ be the sum of weights of vertices in a minimum weight vertex cover $S_{OPT}$ of $G$. Then $\sum_{v \in S_{\delta}}{w(v)} \leq (2+\epsilon)OPT$. In particular, if $S_{\delta}$ is a vertex cover then it is a $(2+\epsilon)$-approximation for MWVC for $G$.
\end{theorem}

\begin{proof}
For every $v \in V$ we have that $w'_{\delta}(v)=w(v)-\tilde{w}_{\delta}(v)$, which implies that $w(v)=w'_{\delta}(v)+\tilde{w}_{\delta}(v)$. For every $v \in S_{\delta}$ it holds that $w'_{\delta}(v) \leq \epsilon' w(v)$, and therefore $w(v)\leq \epsilon' w(v)+\tilde{w}_{\delta}(v)$. Put otherwise, for every $v \in S_{\delta}$ we have $w(v)\leq (1/(1-\epsilon'))\tilde{w}_{\delta}(v)$. This gives:
\begin{eqnarray*}
\sum_{v \in S_{\delta}}{w(v)} &\leq& \frac{1}{(1-\epsilon')}\sum_{v \in S_{\delta}}{\tilde{w}_{\delta}(v)}\\
&\leq& \frac{1}{(1-\epsilon')}\sum_{v \in S_{\delta}}{\sum_{e: v \in e}{\delta(e)}} \\
&\leq& \frac{1}{(1-\epsilon')}\sum_{v \in V}{\sum_{e: v \in e}{\delta(e)}}\\
&\leq& \frac{1}{(1-\epsilon')} \cdot 2 \sum_{e \in E}{\delta(e)}.
\end{eqnarray*}
The above is at most $(2/(1-\epsilon')) OPT$ because $OPT \geq \sum_{e \in E}{\delta(e)}$. To see why $OPT \geq \sum_{e \in E}{\delta(e)}$, associate each edge $e$ with its endpoint $v_e$ in $S_{OPT}$ (choose an arbitrary endpoint if both are in $S_{OPT}$). The weight $w(v)$ of each $v \in S_{OPT}$ is at least $\sum_{e: v_e=v}{\delta(e)}$, because it is at least $\sum_{e: v \in e}{\delta(e)}$. Hence, $OPT = \sum_{v \in S_{OPT}}{w(v)} \geq \sum_{v \in S_{OPT}}{\sum_{e: v_e=v}{\delta(e)}} = \sum_{e \in E}{\delta(e)}$.
Hence the sum of weights in $S_{\delta}$ is at most a factor $2/(1-\epsilon')$ larger than $OPT$. Since $\epsilon'= \epsilon/(2+\epsilon)$, we have that $2/(1-\epsilon') = (2-2\epsilon'+2\epsilon')/(1-\epsilon') = 2(1 + \epsilon'/(1-\epsilon')) = 2 + \epsilon$, which completes the proof.
\end{proof}

In the next section, we show how to implement efficiently in a distributed setting an algorithm that finds a function $\delta$ that is $G$-valid, for which the set $S_{\delta}$ is a vertex cover. This immediately gives a distributed $(2+\epsilon)$-approximation for MWVC.

\section{A fast distributed implementation}
\label{sec:dist}
Our goal in this section is to find a $G$-valid function $\delta$ such that $S_{\delta}$ is a vertex cover. Since every vertex knows whether it is in $S_{\delta}$, this immediately gives a distributed $(2+\epsilon)$-approximation algorithm for MWVC. Our algorithm is deterministic and requires for every vertex $v$ only $O(\log{d(v)}/\epsilon'\log\log{d(v)})$ rounds, where $d(v)$ is the degree of $v$ in $G$, or $O(1/\epsilon')$ if $d(v)\leq 16$. Here $\epsilon'= \epsilon/(2+\epsilon)$ where $\epsilon=O(1)$, which means that $\epsilon'=\Theta(\epsilon)$.
For clarity of presentation, in this section we describe an implementation for the LOCAL model. In Section~\ref{sec:discussion} we show how this can be easily be adapted to the CONGEST model in which the message size is limited to $O(\log{n})$ bits, provided that the initial weights of the vertices and the ratio between the maximal and minimal weights can be expressed by $O(\log{n})$ bits.

In our algorithm, each vertex converges to agreeing with each of its neighbors on a function $\delta$ that is $G$-valid, by iterating the process of decreasing the weight of neighbors by the same amount, until either its weight is below a small fraction of its original weight or it has no more neighbors in the graph induced by the vertices that remain so far. This would imply that the set of vertices whose weight decreased below the above threshold is a vertex cover, and by Theorem~\ref{thm:lr} its weight is a $(2+\epsilon)$-approximation to the weight of a minimal vertex cover.

\paragraph{Overview of Algorithm~\ref{alg}:} The algorithm consists of iterations, each of which has a constant number of communication rounds. Each vertex $v$ splits its current weight $w_i(v)$ into two amounts. The first amount is $\Vault(v)$, which is equal to its threshold $\epsilon' w_0(v)$, and the second amount is $\Bank(v)$, which contains the rest of the weight $w_i(v) - \Vault(v)$. Notice that $\epsilon'<1$ because $\epsilon'= \epsilon/(2+\epsilon)$ and therefore these amounts are well-defined.

In each iteration, vertex $v$ sends a $\Deal_i(v,u)$ request to its neighbor $u$, which is the amount in its $\Vault(v)$ divided by the current number of neighbors of $v$. This guarantees that any weight decrease that results for $v$ from this part does not exceed its total remaining weight. The second amount is used to respond to $\Deal_i(u,v)$ requests from its neighbors. The vertex $v$ processes these requests one by one in any arbitrary order, and responds with the amount $\Budget_i(v,u)$ which is the largest amount by which $v$ can currently decrease its weight, and no more than the request $\Deal_i(u,v)$. This amount is decreased from $\Bank(v)$, and hence it is also guaranteed that decreasing this amount does not exceed the total remaining weight.

Once the weight of $v$ reaches its threshold, $v$ completes its algorithm, returning $\InCover$ after notifying its neighbors of this fact. If the weight is still above the threshold then $v$ only removes its edges to neighbors that notified they are returning $\InCover$. This gives that each edge has at least one endpoint returning $\InCover$, and hence the set of all such vertices is a vertex cover, and by Theorem~\ref{thm:lr} its weight is a $(2+\epsilon)$-approximation to the weight of a minimal vertex cover. The analysis of the number of rounds is based on the observation that for each of the neighbors $u$ of $v$, either it decreases its weight by responding to $\Deal_i(v,u)$ with the entire amount and thereby contributes to decreasing the weight of $v$ by this amount, or it does not have the required budget in its $\Bank(u)$, in which case it contributes to decreasing the number of neighbors of $v$ by $1$.

We proceed by the full pseudocode, followed by an explicit analysis.

\RestyleAlgo{boxruled}
\LinesNumbered
\begin{algorithm}[htbp]
	\label{alg}
	\caption{A distributed $(2+\epsilon)$-approximation algorithm for MWVC, code for vertex $v$.}
	$w_0(v) = w(v)$\\
	$d_0(v) = d(v)$\\
	$N_0(v) = N(v)$\\
	$i=0$\\
	$\Vault(v) = \epsilon' w_0(v)$\\
	\While{$true$}
	{
		$\Bank(v) = w_i(v) - \Vault(v)$\\
		$w_{i+1}(v) = w_i(v) $\\
		$N_{i+1}(v) = N_i(v)$\\
		\ForEach{$u \in N_i(v)$}
		{		
			$\Deal_i(v,u) = \Vault(v)/d_i(v)$\\
			Send $\Deal_i(v,u)$ to $u$\\
			Let $\Budget_i(u,v)$ be the response from $u$\\
			$w_{i+1}(v) = w_{i+1}(v) - \Budget_i(u,v)$\\
			\If{$\Budget_i(u,v) < \Deal_i(v,u)$}
			{
				$N_{i+1}(v) = N_{i+1}(v)\setminus\{u\}$
			}
		}
		Let $u_1\dots u_{d_i(v)}$ be an order of $N_i(v)$\\
		\ForEach{$k=1,\dots,d_i(v)$}
		{
			Let $\Deal_i(u_k,v)$ be received from $u_k \in N_i(v)$\\
			$\Budget_i(v,u_k) = \min\{\Deal_i(u_k,v), \Bank(v)-\sum_{t=1}^{k-1}{\Budget_i(v,u_t)} \}$\\
			Send $\Budget_i(v,u_k)$ to $u_k$
		}
		$\Bank(v) = \Bank(v) - \sum_{k=1}^{d_i(v)}{\Budget_i(v,u_k)}$\\
		$w_{i+1}(v) = w_{i+1}(v) - \sum_{k=1}^{d_i(v)}{\Budget_i(v,u_k)}$\\
		$d_{i+1}(v) = |N_{i+1}(v)|$\\
		$i=i+1$\\
		\If{$w_i(v) \leq \epsilon' w_0(v)$}
		{
            		Send $(v,cover)$ to all neighbors\\
            		Return $\InCover$\\
        		}
        		\ForEach{$(u,cover)$ received from $u \in N_i(v)$}
        		{
            		$N_i(v) = N_{i}(v)\setminus\{u\}$\\
            		$d_i(v) = d_i(v) - 1$
        		}
        		\If{$d_i(v)=0$}
        		{
            		Return $\NotInCover$
        		}
    	}	
\end{algorithm}

\begin{theorem}
\label{thm:alg}
For every $\epsilon=O(1)$, Algorithm~\ref{alg} is a deterministic distributed $(2+\epsilon)$-approximation algorithm for MWVC in which each vertex $v$ with degree $d(v)$ completes in $O(1/\epsilon)$ rounds if $d(v)\leq 16$, and in $O(\log{d(v)}/\epsilon\log\log{d(v)})$ rounds otherwise.
\end{theorem}
We split the proof into two parts, showing correctness and complexity separately. We begin by showing correctness in the following lemma. Essentially, we show that the algorithm finds a function $\delta$ that is $G$-valid and for which $S_{\delta}$ is a vertex cover.

\begin{lemma}
\label{lemma:approx}
Algorithm~\ref{alg} is a deterministic distributed $(2+\epsilon)$-approximation algorithm for MWVC.
\end{lemma}
\begin{proof}
We first show that Algorithm~\ref{alg} is a $(2+\epsilon)$-approximation algorithm for MWVC. That is, we claim that the set $C=\{v \in V~|~ v \mbox{ outputs } \InCover\}$ is a vertex cover, and that $\sum_{v \in C}{w(v)} \leq (2+\epsilon)OPT$, where $OPT = \sum_{v \in S_{OPT}}{w(v)}$ for some optimal vertex cover $S_{OPT}$. For this, we show that the sum of amounts deducted by neighbors can be used to define a $G$-valid function over the edges. This will be exactly the function according to which the vertices decide whether to output $\InCover$ or $\NotInCover$.

For every $e = \{v,u\} \in E$ and every $i = 0,1 \dots$, let $\delta_i(e)=\Budget_i(u,v)+\Budget_i(v,u)$. Let $\delta(e)=\sum_{i=0,1,\dots}{\delta_i(e)}$. We claim that $\delta$ is $G$-valid, i.e., for every vertex $v$ it holds that $\sum_{e: v \in e}{\delta(e)} \leq w(v)$. Let $j$ be the value of $i$ when $v$ returns, that is, $v$ participates in iterations $i=0,\dots,j-1$. For each iteration $i=0,\dots j-1$ it holds that
$$\sum_{u_k \in N_i(v)}{\Budget_i(u_k,v)} \leq \sum_{u_k \in N_i(v)}{\Vault(v)/d_i(v)} = \Vault(v),$$
where $N_i(v)=\{u_1,\dots,u_{d_i(v)}\}$ is the set of neighbors of $v$ at the beginning of iteration $i$.
Further, since for $u_k \in N_i(v)$, $\Budget_i(v,u_k) = \min\{\Deal_i(u_k,v),\Bank(v)-\sum_{t=1}^{k-1}{\Budget_i(v,u_t)} \}$, we have that
$$\sum_{u_k \in N_i(v)}{\Budget_i(v,u_k)} \leq \Bank(v).$$

Since $\Bank(v) = w_i(v) - \Vault(v)$ it holds that $\sum_{e=\{v,u_k\}: u_k \in N_i(v)}{\delta_i(e)} \leq w_i(v)$. Since $w_{i+1}(v)= w_i(v)-\sum_{e: v \in e}{(\Budget_i(u,v)+\Budget_i(v,u))}$, we have that $w_{i+1}(v) = w_i(v)-\sum_{e: v \in e}{\delta_i(e)} \geq 0$.
This gives that $w(v) = \sum_{i=0}^{j-1}{(w_i(v)-w_{i+1}(v))}+w_j(v) = \sum_{i=0}^{j-1}{\sum_{e: v \in e}{\delta_i(e)}}+w_j(v) \geq 0$, and hence $w(v)-\sum_{e: v \in e}{\delta(e)} = w_j(v) \geq 0$.

This proves that $\delta$ is $G$-valid, which gives that for  $C=\{v \in V~|~ v \mbox{ outputs } \InCover\}$ it holds that $\sum_{v \in C}{w(v)} \leq (2+\epsilon)OPT$, where $OPT = \sum_{v \in S_{OPT}}{w(v)}$ for some optimal vertex cover $S_{OPT}$, by Theorem~\ref{thm:lr}. This is because a vertex $v$ outputs $\InCover$ at the end of iteration $i=j-1$ if and only if $w_{j}(v) \leq \epsilon' w_0(v)$. It remains to show that $C$ is a vertex cover. To see why, consider an edge $e=\{v,u\} \in E$. We claim that if $u,v$ have both returned by the end of iteration $i$, then at least one of them is in $C$. This is because otherwise $d_{i+1}(v),d_{i+1}(u) \geq 1$, which implies that both have not returned yet. This completes the proof that $C$ is indeed a $(2+\epsilon)$-approximation for MWVC.
\end{proof}

It remains to bound the number of rounds. We do so in the following lemma, in which we show that in each iteration either enough weight is reduced or enough neighbors enter the vertex cover.
\begin{lemma}
\label{lemma:complexity}
In Algorithm~\ref{alg}, each vertex $v$ with degree $d(v)$ completes in $O(1/\epsilon)$ rounds if $d(v)\leq 16$, and in $O(\log{d(v)}/\epsilon\log\log{d(v)})$ rounds otherwise.
\end{lemma}
\begin{proof}
Let $K_v>1$ be a parameter to be chosen later. Let $i$ be an iteration at the beginning of which a vertex $v \in V$ has not yet returned. We claim that either
$d_{i+1}(v) \leq d_i(v)/K_v$ or $w_{i+1}(v) \leq w_i(v)-\epsilon'w_0(v)/K_v$. To see why, suppose $d_{i+1}(v) > d_i(v)/K_v$. This means that for at least $\ceil{d_i(v)/K_v}$ vertices $u \in N_i(v)$, it holds that $\Budget_i(u,v) = \Deal_i(v,u)$, and hence
\begin{eqnarray*}
w_{i+1}(v) &\leq& w_i(v) - \left\lceil\frac{d_i(v)}{K_v}\right\rceil\cdot\frac{\Vault(v)}{d_i(v)}\\
&\leq& w_i(v) - \frac{d_i(v)}{K_v}\cdot\frac{\epsilon' w_0(v)}{d_i(v)}\\
&\leq& w_i(v) - \epsilon' w_0(v)/K_v.
\end{eqnarray*}

Next, we claim that $v$ returns after at most $K_v/\epsilon' + \log{d(v)}/\log{K_v}$ iterations of the algorithm. This is because at most $\log_{K_v}{d(v)}=\log{d(v)}/\log{K_v}$ of the iterations $i$ can be such that $d_{i+1}(v) \leq d_i(v)/K_v$ (since $v$ returns when $d_i(v)=0$), and at most $K_v/\epsilon'$ iterations $i$ can be such that $w_{i+1}(v) \leq w_i(v)-\epsilon'w_0(v)/K_v$ (since $v$ returns when $w_i(v) \leq \epsilon'w_0(v)$).

Finally, we set $K_v$ as follows. If $d(v)\leq 16$ we set $K_v=d(v)+1$. This guarantees $K_v>1$ (an isolated vertex simply outputs \NotInCover) and gives $O(1/\epsilon)$ rounds for $v$ to complete.

Otherwise, we set $K_v=\log{d(v)}/\log{\log{d(v)}}$. Since $d(v)>16$, it holds that $K_v$ is well defined (as $\log\log{d(v)}>1$) and that $K_v>1$. It also holds that $\log{K_v}>1$ which is used in what follows.
This gives that vertex $v$ returns after at most $j$ iterations, where
\begin{eqnarray*}
j &\leq&  K_v/\epsilon' + \log{d(v)}/\log{K_v}\\
&=& \frac{\log{d(v)}}{\epsilon'\log{\log{d(v)}}} +\log{d(v)}/\log{K_v}\\
&=& \frac{\log{d(v)}}{\epsilon'\log{\log{d(v)}}} + \frac{\log{d(v)}}{\log{(\log{d(v)}/\log{\log{d(v)}})}}\\
&=& \frac{\log{d(v)}}{\epsilon'\log{\log{d(v)}}} + \frac{\log{d(v)}}{\log\log{d(v)}-\log\log{\log{d(v)}}}\\
&\leq& O\left(\frac{\log{d(v)}}{\epsilon\log\log{d(v)}}\right),
\end{eqnarray*}
where the last inequality follows because $\epsilon'= \epsilon/(2+\epsilon)$ (and since $\epsilon$ is at most $O(1)$ and so $\epsilon'=\Theta(\epsilon)$) and $\log\log{d(v)}$ dominates $\log\log\log{d(v)}$, completing the proof.
\end{proof}

Theorem~\ref{thm:alg} follows directly from Lemmas~\ref{lemma:approx} and ~\ref{lemma:complexity}.

\section{Adaptation to the CONGEST model}
\label{sec:discussion}
Our algorithm is described for the LOCAL model, but can be easily adapted to the CONGEST model in which the message size is limited to $O(\log{n})$ bits, provided that the initial weights of the vertices and the ratio between the maximal and minimal weights can be expressed by $O(\log{n})$ bits. In order to accommodate $O(\log{n})$-bit messages, we slightly modify the messages that are sent as follows. First, in an initial round, each vertex $v$ sends $w_0(v)$ to all of its neighbors. Then, instead of sending $\Deal_i(v,u)$ to neighbor $u$ in some iteration $i$, vertex $v$ only needs to send $d_i(v)$ to its neighbor $u$ and $u$ can locally compute $\Deal_i(v,u)=\Vault(v)/d_i(v)$ since all vertices know the value of $\epsilon$ as part of their algorithm.

Second, we need to handle the messages of type $\Budget_i(v,u)$. In general, this amount can be an arbitrary fraction which might not fit in $O(\log{n})$ bits. However, we notice that we can avoid sending this explicit amount. To do this, we slightly modify $\Vault(v)$ to be $\epsilon'w_0(v)/2$. Then, upon receiving a $\Deal_i(u,v)$ message, if $\Budget_i(v,u)=\Deal_i(u,v)$ then vertex $v$ replies with a predefined message \emph{accept}, and otherwise, $v$ responds with the maximal integer $t$ such that $t\epsilon'w_0(v)/2 \leq \Budget_i(v,u)$. The amount $t\epsilon'w_0(v)/2$ can be locally computed by $u$, and $u$ can infer that $v$ returns $\InCover$. This is because the remainder of weight in vertex $v$ will be another value of at most $\epsilon'w_0(v)/2$ on top of the at most $\epsilon'w_0(v)/2$ value which might remain in $\Vault(v)$, summing to no more than $\epsilon'w_0(v)$, as needed.

\section{Discussion of~\cite{KMW10}}
\label{sec:refutation}
The main result of~\cite{KMW10} is Theorem $9$, which states the following:

\paragraph{Theorem 9 from~\cite{KMW10}.}
\textit{
For every constant $\epsilon>0$, there are graphs $G$, such that in $k$ communication rounds, every distributed algorithm for the minimum vertex cover problem on $G$ has approximation ratios at least $$\Omega\left(n^{\frac{1/4-\epsilon}{k^2}}\right) \mbox{ and } \Omega\left(\Delta^{\frac{1-\epsilon}{k+1}}\right),$$
where $n$ and $\Delta$ denote the number of nodes and the highest degree in $G$, respectively.
}

~\\

The argument in~\cite{KMW10} is that in order for the above approximation factors to be constant, the number of rounds, $k$, has to be $\Omega(\sqrt{\log{n}})$ and $\Omega(\log{\Delta})$, respectively.

However, we argue that the above lower bounds only hold under the conditions that $k = O((\log{n})^{1/3})$ and $k = O(\sqrt{\log{\Delta}})$, respectively. This means that they cannot be applied to $k=\Theta(\sqrt{\log{n}})$ or $k=\Theta(\log{\Delta})$, and therefore do not imply the claimed bounds for constant approximation factors.

To justify our claim, we elaborate upon the proof of the theorem.
Previous lemmas in the paper\footnote{We refer the reader to~\cite{KMW10} for exact details.} show that the approximation factor of any $k$-round algorithm is $\Omega(\delta)$, where $\delta$ satisfies the following two constraints\footnote{We use the notation $\delta$ as this is the notation in~\cite{KMW10}. Notice that it is unrelated to the function $\delta$ that we use in our framework in previous sections of this paper.}. First, it holds that $n \leq 2^{2k^3+4k}\delta^{4k^2}$ and second, it holds that $\Delta = 2^{k(k+1)/2}\delta^{k+1}$.

The first constraint implies that
$$\delta \geq \frac{n^{1/4k^2}}{2^{(2k^3+4k)/4k^2}} = \frac{n^{1/4k^2}}{n^{(2k^3+4k)/4k^2\log{n}}} = n^{1/4k^2-(2k^3+4k)/4k^2\log{n}}.$$
Hence, in order to deduce that $\delta = \Omega(n^{\frac{1/4-\epsilon}{k^2}})$, it needs to hold that $(2k^3+4k)/4k^2\log{n}\leq \epsilon/k^2$. However, for this to happen, it must be that $2k^3+4k \leq 4\epsilon\log{n}$, and in particular $k$ has to be within $O((\log{n})^{1/3})$.

The second constraint implies that
$$\delta = \frac{\Delta^{1/(k+1)}}{2^{k/2}} = \frac{\Delta^{1/(k+1)}}{\Delta^{k/2\log{\Delta}}} = \Delta^{1/(k+1)-k/2\log{\Delta}}.$$
Hence, in order to deduce that $\delta = \Omega(\Delta^{\frac{1-\epsilon}{k+1}})$, it needs to hold that $k/2\log{\Delta}\leq \epsilon/(k+1)$. However, for this to happen, it must be that $k(k+1) \leq 2\epsilon\log{\Delta}$, and in particular $k$ has to be within $O(\sqrt{\log{\Delta}})$.

We emphasize again that this last step in the proof of the lower bound is different in the previous version~\cite{KMW04}, and hence we do not suggest that there is a flaw in~\cite{KMW04}.

\paragraph{Acknowledgements:} We thank Seri Khoury and Dror Rawitz for many discussions and helpful suggestions.

\bibliographystyle{alpha}
\bibliography{paper}

\end{document}